\documentclass[11pt]{article}

\usepackage[margin=1in]{geometry}
\usepackage{graphicx}
\usepackage[pdftex,colorlinks=true,linkcolor=blue,citecolor=blue,urlcolor=black]{hyperref}
\usepackage{amsmath, amsthm, amssymb}
\usepackage{subfigure}
\usepackage{comment}
\usepackage{url}
\usepackage{xcolor}
\usepackage[boxruled,lined]{algorithm2e}

\newcommand{\F}{\mathbb{F}}

\newcommand{\ket}[1]{| #1 \rangle}

\DeclareMathOperator{\Tr}{Tr}

\newcommand{\be}{\begin{equation}}
\newcommand{\ee}{\end{equation}}
\newcommand{\bea}{\begin{eqnarray}}
\newcommand{\eea}{\end{eqnarray}}
\newcommand{\bes}{\begin{equation*}}
\newcommand{\ees}{\end{equation*}}
\newcommand{\beas}{\begin{eqnarray*}}
\newcommand{\eeas}{\end{eqnarray*}}

\newtheorem{thm}{Theorem}
\newtheorem*{thm*}{Theorem}

\newtheorem{lem}[thm]{Lemma}
\newtheorem*{lem*}{Lemma}
\newtheorem{prop}[thm]{Proposition}

\theoremstyle{definition}
\newtheorem{dfn}{Definition}

\begin{document}

\title{The quantum query complexity of learning multilinear polynomials}

\author{Ashley Montanaro\footnote{Centre for Quantum Information and Foundations, Department of Applied Mathematics and Theoretical Physics, University of Cambridge, UK; {\tt am994@cam.ac.uk}.}}

\maketitle

\begin{abstract}
In this note we study the number of quantum queries required to identify an unknown multilinear polynomial of degree $d$ in $n$ variables over a finite field $\F_q$. Any bounded-error classical algorithm for this task requires $\Omega(n^d)$ queries to the polynomial. We give an exact quantum algorithm that uses $O(n^{d-1})$ queries for constant $d$, which is optimal. In the case $q=2$, this gives a quantum algorithm that uses $O(n^{d-1})$ queries to identify a codeword picked from the binary Reed-Muller code of order $d$.
\end{abstract}


\section{Introduction}

A central problem in computational learning theory is to determine the complexity of identifying an unknown function of a certain type, given access to that function via an oracle. We say that a class $\mathcal{F}$ of functions can be {\em learned} using $t$ queries if any function $f \in \mathcal{F}$ can be identified with $t$ uses of $f$ (perhaps allowing some probability of error). It is known that some classes of functions can be learned more efficiently by quantum algorithms than is possible classically. In particular, one of the earliest results in the field of quantum computation is that the class of linear functions $\F_2^n\rightarrow \F_2$ (also known as Hadamard codewords) can be learned using a single quantum query~\cite{bernstein97}, whereas $\Omega(n)$ queries are required classically. Here we generalise this result to quantum learning of {\em multilinear} functions over general finite fields.

Let $\F_q$ denote the finite field with $q=p^r$ elements for some prime $p$. Every function $f:\F_q^n \rightarrow \F_q$ can be represented as a polynomial in $n$ variables over $\F_q$. $f$ is said to be a degree $d$ polynomial if it can be written as a polynomial whose every term is of total degree at most $d$. For example, the function $f:\F_5^3 \rightarrow \F_5$ defined by $f(x) = 2 x_1 + 4 x_1^2 x_2 + x_1 x_2 x_3$ is a degree 3 polynomial. The set of polynomials of degree $d$ in $n$ variables over $\F_q$ is known as the (generalised) Reed-Muller code of order $d$ over $\F_q$.

We say that a degree $d$ polynomial $f:\F_q^n \rightarrow \F_q$ is multilinear if it can be written as
\[ f(x) = \sum_{S \subseteq [n],|S|\le d} \alpha_S \prod_{i \in S} x_i \]
for some coefficients $\alpha_S \in \F_q$, where $[n]$ denotes the set $\{1,\dots,n\}$. Note that in the case $S = \emptyset$ we define $\prod_{i \in S} x_i = 1$. For example, any multilinear polynomial of degree 3 can be written as
\[ f(x) = \alpha_{\emptyset} + \sum_{i} \alpha_{\{i\}} x_i + \sum_{i<j} \alpha_{\{i,j\}} x_i x_j + \sum_{i<j<k} \alpha_{\{i,j,k\}} x_i x_j x_k. \]
Technically, such functions are {\em multiaffine} rather than multilinear, as they are affine in each variable; however, we use the ``multilinear'' terminology for consistency with prior work. In particular, note that in this terminology, linear functions $f:\F_q^n \rightarrow \F_q$ (i.e.\ functions such that $f(x+y) = f(x) + f(y)$) are equivalent to degree 1 multilinear polynomials with no constant term. In the important special case $q=2$ (Boolean functions), every function $f:\F_2^n \rightarrow \F_2$ is multilinear.

Given the ability to query a multilinear degree $d$ polynomial $f$ on arbitrary $x \in \F_q^n$, we would like to determine ({\em learn}) $f$ using the smallest possible number of queries. A straightforward classical algorithm can solve this problem by querying $f(x)$ for all strings $x \in \F_q^n$ that contain only 0 and 1, and such that $|x| \le d$. (We write $|x|$ for the Hamming weight of $x\in \F_q^n$, i.e.\ the number of non-zero components.) To see this, first consider the special case where for some $k$, $\alpha_S = 0$ for all $S$ such that $|S| < k$. Then knowing $f(x)$ for all $x$ of the above form such that $|x|=k$ is sufficient to determine all of the degree $k$ coefficients of $f$ (note that this relies on $f$ being multilinear). More generally, let $f_k$ denote the degree $k$ part of $f$, i.e.
\[ f_k(x) = \sum_{S \subseteq [n],|S|=k} \alpha_S \prod_{i \in S} x_i.\]
For any $k$, once $f_{\ell}$ is known for all $\ell \le k$, the degree $k+1$ coefficients can be determined from the inputs of Hamming weight $k+1$: whenever $f$ is queried on $x$, subtract $\sum_{\ell=0}^k f_{\ell}(x)$ from the result to simulate that $\alpha_S=0$ for all $S$ such that $|S|\le k$. The algorithm can therefore learn $f$ with certainty using $1 + n + \binom{n}{2}+\binom{n}{3}+\dots+\binom{n}{d}$ queries, which is $O(n^d)$ for constant $d$. In the special case of functions $f:\F_2^n \rightarrow \F_2$, all polynomials are multilinear. This implies that the class of all degree $d$ polynomials $f:\F_2^n \rightarrow \F_2$ can be learned using $O(n^d)$ queries.


It is also easy to see that the above algorithm is exactly optimal in an information-theoretic sense. As the number of distinct multilinear degree $d$ polynomials of $n$ variables over $\F_q$ is equal to
\[ q^{1 + n + \binom{n}{2} + \binom{n}{3} + \dots + \binom{n}{d}}, \]
and as a classical query to $f$ only provides $\log_2 q$ bits of information, any classical algorithm must make $1 + n+\binom{n}{2}+\binom{n}{3}+\dots+\binom{n}{d} = \Omega(n^d)$ queries to $f$ in order to identify it with certainty. A similar bound can be proven for bounded-error algorithms. Indeed, let $f$ be picked uniformly at random, and consider an algorithm (without loss of generality deterministic) that makes at most $c$ queries to $f$ before it outputs an answer. Such an algorithm can output the correct answer for at most $q^c$ functions $f$, and hence succeeds with probability at most $q^{c-(1 + n + \binom{n}{2} + \binom{n}{3} + \dots + \binom{n}{d})}$.

Using similar techniques, one can find a lower bound for {\em quantum} query algorithms~\cite{hoyer05}. In the standard quantum query model, the algorithm accesses $f$ via the unitary operation $O_f \ket{x}\ket{y} = \ket{x}\ket{y + f(x)}$, where $x \in \F_q^n$, $y \in \F_q$. We formalise a lower bound on the number of queries required to identify $f$ in this model as the following proposition\footnote{In the case $q = 2$, this lower bound can also be obtained from independent results of Farhi et al~\cite{farhi99} and Servedio and Gortler~\cite{servedio01}.}.

\begin{prop}
Let $f:\F_q^n \rightarrow \F_q$ be a multilinear degree $d$ polynomial over $\F_q$. Then any quantum query algorithm which learns $f$ with bounded error must make $\Omega(n^{d-1})$ queries to $f$.
\end{prop}

\begin{proof}
Each query can be seen as a round of a communication process, where in each round the algorithm sends the registers $\ket{x}$ and $\ket{y}$ to the oracle, using $(n + 1) \log_2 q$ qubits of communication; the oracle then performs the map $\ket{x}\ket{y} \mapsto \ket{x}\ket{y + f(x)}$ and returns the registers to the algorithm. Let $f$ be picked uniformly at random from the set of degree $d$ multilinear polynomials, let $X$ be the corresponding random variable, and let $Y$ be the random variable corresponding to the function which is output by the algorithm. By Holevo's theorem \cite{holevo73} (see also~\cite{cleve98}), after $r$ rounds of communication, the mutual information between $X$ and $Y$ satisfies the upper bound
\[ I(X:Y) \le 2r (n+1) \log_2 q. \]
On the other hand, Fano's inequality~\cite{cover06} states that the probability $P_e$ of identifying $f$ incorrectly satisfies the lower bound
\[ P_e \ge 1 - \frac{I(X:Y)+1}{\log_2\left(q^{1 + n+\binom{n}{2}+\binom{n}{3}+\dots+\binom{n}{d}}\right)}, \]
which thus implies that 
\[ P_e \ge 1 - \frac{2r (n+1) + 1/\log_2 q}{1 + n+\binom{n}{2}+\binom{n}{3}+\dots+\binom{n}{d}}. \]
For this quantity to be upper bounded by a constant, we must have $r = \Omega(n^{d-1})$.
\end{proof}

The main result of this note is that this asymptotic scaling can actually be achieved.

\begin{thm}
\label{thm:main}
Let $f:\F_q^n \rightarrow \F_q$ be a multilinear degree $d$ polynomial over $\F_q$. Then there is an exact quantum algorithm which learns $f$ with certainty using $1 + \sum_{i=1}^d 2^{i-1} \binom{n}{i-1}$ queries to $f$, which is $O(n^{d-1})$ for constant $d$.
\end{thm}

The case $d=1$, $q=2$ of this result was previously proven by Bernstein and Vazirani~\cite{bernstein97}, while a bounded-error quantum algorithm using $O(n)$ queries for the case $d=2$, $q=2$ was more recently given by R\"otteler~\cite{roetteler09}; by contrast, the algorithm given here is exact and works for all $d$ and all fields $\F_q$. In related work, a quantum algorithm for estimating quadratic forms over the reals using $O(n)$ queries had previously been given by Jordan \cite[Appendix D]{jordan08}.


\section{Proof of Theorem \ref{thm:main}}

The only quantum ingredient we will need to prove Theorem \ref{thm:main} is the following lemma, which is implicit in \cite{beaudrap02,vandam06} and is a simple extension of the Bernstein-Vazirani algorithm \cite{bernstein97} for identifying linear functions over $\F_2$.

\newcounter{linlem}\setcounter{linlem}{\value{thm}}

\begin{lem}[\cite{beaudrap02,vandam06}]
\label{lem:linear}
Let $f:\F_q^n \rightarrow \F_q$ be linear, and let $g:\F_q^n \rightarrow \F_q$ be the function $g(x) = f(x) + \beta$ for some constant $\beta \in \F_q$. Then $f$ can be determined exactly using one quantum query to $g$.
\end{lem}

For completeness, we give a full proof of Lemma \ref{lem:linear} in Appendix \ref{sec:linproof}.

We will derive a quantum algorithm to learn an unknown multilinear degree $d$ polynomial $f$ by introducing a {\em linear} function $f_S$ of $n$ variables which can be produced using a relatively small number of queries to $f$, and from which $f$ can be determined using Lemma \ref{lem:linear}. This technique is somewhat similar to the approach used to learn quadratic polynomials with bounded error in the work~\cite{roetteler09}. A related function was previously used by Kaufman and Ron \cite{kaufman06} to produce an efficient classical {\em tester} for low-degree polynomials over finite fields.

For any $k$-subset $S \subseteq [n]$, let $S_j$ denote the $j$'th element of $S$, where $S$ is considered as an increasing sequence of integers. For $i \in [n]$, let $e_i$ denote the $i$'th element in the standard basis for the vector space $\F_q^n$. For any $f:\F_q^n \rightarrow \F_q$ and any subset $S \subseteq [n]$, define the function $f_S:\F_q^n \rightarrow \F_q$ as follows:
\[ f_S(x) = \sum_{\beta_1,\dots,\beta_k \in \{0,1\}} (-1)^{k - \sum_{i=1}^{k} \beta_i}\,f\left(x + \sum_{j=1}^k \beta_j e_{S_j}\right), \]
where the inner sum is over $\F_q^n$ and the outer sum is over $\F_q$. For example, for $S = \{1,2\}$, $f_S(x) = f(x) - f(x + e_1) - f(x + e_2) + f(x + e_1 + e_2)$. When $q=2$, $f_S(x)$ sums $f$ over the affine subspace of $\F_2^n$ positioned at $x$ and spanned by $\{e_i:i\in S\}$. It is clear that a query to $f_S$ can be simulated using $2^k$ queries to $f$. One way of understanding $f_S$ is in terms of {\em discrete derivative} operators. If we define the discrete derivative of $f$ in direction $i \in [n]$ as $(\Delta_i f)(x) = f(x + e_i) - f(x)$, then $f_S(x) = (\Delta_{S_1} \Delta_{S_2}  \dots \Delta_{S_k} f)(x)$. In other words, $f_S$ is the function obtained by taking the derivative of $f$ with respect to all of the variables in $S$.

We will be interested in querying $f_S$ for sets $S$ of size $d-1$. In this case, we have the following characterisation for multilinear polynomials $f$.

\newcounter{fslem}\setcounter{fslem}{\value{thm}}

\begin{lem}
\label{lem:fs}
Let $f:\F_q^n \rightarrow \F_q$ be a multilinear polynomial of degree $d$ with expansion
\[ f(x) = \sum_{T \subseteq [n],|T| \le d} \alpha_T \prod_{i \in T} x_i. \]
Then, for any $S$ such that $|S|=d-1$,
\[ f_S(x) = \alpha_S + \sum_{k \notin S} \alpha_{S \cup \{k\}} x_k. \]
\end{lem}

Lemma \ref{lem:fs} follows easily from expressing $f_S$ in terms of discrete derivatives; we also give a simple direct proof in Appendix \ref{sec:fsproof}. We are now ready to describe a quantum algorithm which uses $f_S$ to learn the degree $d$ component of $f$.

\begin{algorithm}[H]
\label{alg:learntop}
\ForEach{$S \subseteq [n]$ such that $|S|=d-1$}{
Use one query to $f_S$ to learn the coefficients $\alpha_{S \cup \{k\}}$, for all $k \notin S$\;
}
Output the function $f_d$ defined by $f_d(x) = \sum_{S \subseteq [n],|S|=d} \alpha_S \prod_{i \in S} x_i$\;
\caption{Learning the degree $d$ component of $f$}
\end{algorithm}

Correctness of this algorithm follows from Lemmas \ref{lem:linear} and \ref{lem:fs}. By Lemma \ref{lem:fs}, for any $S$ such that $|S|=d-1$, knowledge of the degree 1 component of $f_S$ is sufficient to determine $\alpha_{S \cup \{k\}}$ for all $k \notin S$. Therefore, knowing the degree 1 part of $f_S$ for all $S \subseteq [n]$ such that $|S|=d-1$ is sufficient to completely determine all degree $d$ coefficients of $f$. By Lemma \ref{lem:linear}, for any $S$ with $|S|=d-1$, the degree 1 component of $f_S$ can be determined with one quantum query to $f_S$. This implies that Algorithm \ref{alg:learntop} completely determines the degree $d$ component of $f$ using $\binom{n}{d-1}$ queries to $f_S$, each of which uses $2^{d-1}$ queries to $f$.

Once the degree $d$ component of $f$ has been learned, $f$ can be reduced to a degree $d-1$ polynomial by crossing out the degree $d$ part whenever the oracle for $f$ is called. That is, whenever the oracle is called on $x$, we subtract $f_d(x)$ from the result (recall $f_d$ is the degree $d$ part of $f$), at no extra query cost. Inductively, $f$ can be determined completely using
\[ 2^{d-1} \binom{n}{d-1} + 2^{d-2} \binom{n}{d-2} + \dots + 2n + 1 + 1 \]
queries; the last query is to determine the constant term $\alpha_{\emptyset}$, which can be achieved by classically querying $f(0^n)$. The number of queries used is therefore $O(n^{d-1})$ for constant $d$, completing the proof of Theorem \ref{thm:main}.


\section*{Acknowledgements}

I would like to thank Salman Beigi for spotting a crucial error in a previous version, and Graeme Mitchison and Tony Short for helpful comments. I would also like to thank two anonymous referees for their suggestions. This work was supported by an EPSRC Postdoctoral Research Fellowship.


\appendix


\section{Quantum learning of linear functions}
\label{sec:linproof}

In order to prove Lemma \ref{lem:linear}, we will use the quantum Fourier transform (QFT) over general finite fields. This was originally defined by de Beaudrap, Cleve and Watrous \cite{beaudrap02} and independently by van Dam, Hallgren and Ip \cite{vandam06}. The QFT over $\F_q$ is defined as the unitary operation
\[ Q_q \ket{x} = \frac{1}{\sqrt{q}} \sum_{y \in \F_q} \omega^{\Tr(xy)} \ket{y}, \]
where $\omega = e^{2 \pi i/p}$ (recall $q=p^r$) and the trace function $\Tr:\F_q \rightarrow \F_p$ is defined by $\Tr(x) := x + x^p + x^{p^2} + \dots + x^{p^{r-1}}$. If $q$ is prime (i.e.\ $r=1$), then of course $\Tr(x) = x$. The trace is linear: $\Tr(x + y) = \Tr(x) + \Tr(y)$ (see \cite{lidl97} for the proof of this and other standard facts about finite fields). This allows the $n$-fold tensor product of QFTs to be written concisely as
\[ Q_q^{\otimes n} \ket{x} = \frac{1}{q^{n/2}} \sum_{y \in \F_q^n} \omega^{\Tr (x \cdot y)} \ket{y}, \]
where $x \cdot y = \sum_{i=1}^n x_i y_i$, the sum being taken over $\F_q$.

For any function $f:\F_q \rightarrow \F_q$, let $U_f$ be the unitary operator that maps $\ket{x} \mapsto \omega^{\Tr(f(x))} \ket{x}$. Given access to $f$, $U_f$ can be implemented using a standard phase kickback trick as follows.

\begin{lem}[\cite{beaudrap02,vandam06}]
\label{lem:kickback}
$U_f$ can be implemented using one query to $f$.
\end{lem}

\begin{proof}
To implement $U_f$, append an ancilla register $\ket{y}$, $y \in \F_q$, in the initial state $\ket{1}$. Apply $Q_q^{-1}$ to this register to produce
\[ \frac{1}{\sqrt{q}} \sum_{y \in \F_q} \omega^{-\Tr(y)} \ket{y}, \]
then apply $O_f$ to both registers (recall $O_f \ket{x}\ket{y} = \ket{x}\ket{y + f(x)}$). For any $x \in \F_q$, the initial state $\ket{x}\ket{1}$ is mapped to
\[
\frac{1}{\sqrt{q}} \ket{x} \sum_{y \in \F_q} \omega^{-\Tr(y)} \ket{y + f(x)} = \frac{1}{\sqrt{q}} \ket{x} \sum_{y \in \F_q} \omega^{-\Tr (y-f(x))} \ket{y}
= \omega^{\Tr f(x)} \ket{x} \frac{1}{\sqrt{q}} \sum_{y \in \F_q} \omega^{-\Tr(y)} \ket{y},
\]
where we use the linearity of the trace function. As the second register is left unchanged by $O_f$, it can be ignored.
\end{proof}

We are now ready to prove Lemma \ref{lem:linear}.

\setcounter{thm}{\value{linlem}}

\begin{lem}[\cite{beaudrap02,vandam06}]
Let $f:\F_q^n \rightarrow \F_q$ be linear, and let $g:\F_q^n \rightarrow \F_q$ be the function $g(x) = f(x) + \beta$ for some constant $\beta \in \F_q$. Then $f$ can be determined exactly using one quantum query to $g$.
\end{lem}

\begin{proof}
First observe that $f$ will be linear if and only if $f(x) = a \cdot x = \sum_{i=1}^n a_i x_i$ for some $a \in \F_q^n$. Create the state
\[ \ket{\psi_g} := \frac{1}{q^{n/2}} \sum_{x \in \F_q^n} \omega^{\Tr (a \cdot x + \beta)} \ket{x} \]
via the technique of Lemma \ref{lem:kickback}, using one query to $g$. Now apply the $n$-fold tensor product of the inverse quantum Fourier transform to produce
\[ (Q_q^{-1})^{\otimes n} \ket{\psi_g} = \frac{1}{q^n} \sum_{x \in \F_q^n} \omega^{\Tr (a \cdot x + \beta)} \sum_{y \in \F_q^n} \omega^{-\Tr(x \cdot y)} \ket{y}
= \frac{1}{q^n} \omega^{\Tr(\beta)} \sum_{y \in \F_q^n} \left( \sum_{x \in \F_q^n} \omega^{\Tr ((a-y) \cdot x)} \right) \ket{y}.
\]
Note that $\beta$ has been relegated to an unobservable global phase, and the sum over $x$ will be zero unless $y=a$, in which case it will equal $q^n$. A measurement in the computational basis therefore yields $a$ with certainty, which suffices to determine $f$.
\end{proof}


\section{Proof of Lemma \ref{lem:fs}}
\label{sec:fsproof}

We finally prove Lemma \ref{lem:fs}, which we restate for convenience.

\setcounter{thm}{\value{fslem}}

\begin{lem}
Let $f:\F_q^n \rightarrow \F_q$ be a multilinear polynomial of degree $d$ with expansion
\[ f(x) = \sum_{T \subseteq [n],|T| \le d} \alpha_T \prod_{i \in T} x_i. \]
Then, for any $S$ such that $|S|=d-1$,
\[ f_S(x) = \alpha_S + \sum_{k \notin S} \alpha_{S \cup \{k\}} x_k. \]
\end{lem}

\begin{proof}
For brevity, write $|\beta| = \sum_{i=1}^{d-1} \beta_i$. Let $\delta_{xy}$ be the Dirac delta function ($\delta_{xy}=1$ if $x=y$, and $\delta_{xy}=0$ otherwise). By the definition of $f_S$, for any $x \in \F_q^n$ we have
\beas
f_S(x) &=& \sum_{\beta_1,\dots,\beta_{d-1} \in \{0,1\}} (-1)^{d-1-|\beta|} \sum_{T \subseteq [n],|T| \le d} \alpha_T \prod_{i \in T} \left(x_i + \sum_{j=1}^{d-1} \beta_j (e_{S_j})_i\right) \\
&=& (-1)^{d-1} \sum_{T \subseteq [n],|T| \le d} \alpha_T \sum_{\beta_1,\dots,\beta_{d-1} \in \{0,1\}} (-1)^{|\beta|} \prod_{i \in T} \left(x_i + \sum_{j=1}^{d-1} \beta_j \delta_{S_ji}\right).
\eeas
Now note that for all $T$ such that $S \nsubseteq T$, the sum over $\beta_1,\dots,\beta_{d-1}$ will equal 0. This is because in this case there must exist an index $j \in [d-1]$ such that $S_j \notin T$, so for this $j$, $\beta_j$ does not appear in the product over $T$. So, after summing over the $\beta_i$ such that $i \neq j$, we are left with the sum $\sum_{\beta_j\in\{0,1\}} (-1)^{\beta_j} K_T$ for some constant $K_T$; this evaluates to 0 for any $K_T$. As $|S|=d-1$ and $|T| \le d$, this implies that we can rewrite $f_S(x)$ as
\beas
f_S(x) &=& (-1)^{d-1}\alpha_S \sum_{\beta_1,\dots,\beta_{d-1} \in \{0,1\}} (-1)^{|\beta|} \prod_{i \in S} \left(x_i + \sum_{j=1}^{d-1} \beta_j \delta_{S_ji}\right)\\
&+& (-1)^{d-1} \sum_{k \notin S} \alpha_{S \cup \{k\}} \sum_{\beta_1,\dots,\beta_{d-1} \in \{0,1\}} (-1)^{|\beta|} \prod_{i \in S \cup \{k\}} \left(x_i + \sum_{j=1}^{d-1} \beta_j \delta_{S_ji}\right)\\
&=& (-1)^{d-1} \alpha_S \sum_{\beta_1,\dots,\beta_{d-1} \in \{0,1\}} (-1)^{|\beta|} \prod_{i=1}^{d-1} \left(x_{S_i} + \beta_i \right)\\
&+& (-1)^{d-1} \sum_{k \notin S} \alpha_{S \cup \{k\}} \sum_{\beta_1,\dots,\beta_{d-1} \in \{0,1\}} (-1)^{|\beta|} x_k \prod_{i=1}^{d-1} \left(x_{S_i} + \beta_i \right)\\
&=& (-1)^{d-1} \left(\prod_{i=1}^{d-1} \left( \sum_{\beta_i\in \{0,1\}} (-1)^{\beta_i} (x_{S_i} + \beta_i) \right) \right) \left( \alpha_S + \sum_{k \notin S} \alpha_{S \cup \{k\}}x_k \right)\\
&=& \alpha_S + \sum_{k \notin S} \alpha_{S \cup \{k\}}x_k
\eeas
as claimed.
\end{proof}



\end{document}